\documentclass[11pt]{article}
\usepackage{times}
\usepackage{amssymb,amsmath,enumerate}
\usepackage{times,bm,mathrsfs}  
\usepackage{amssymb}
\usepackage{amsfonts}
\usepackage{mathrsfs}
\usepackage{color}
\usepackage{graphicx,xspace,color}
\usepackage{amscd}
\usepackage{epsfig}

\definecolor{Red}{rgb}{1,0,0}
\definecolor{Blue}{rgb}{0,0,1}
\definecolor{Olive}{rgb}{0.41,0.55,0.13}
\definecolor{Green}{rgb}{0,1,0}
\definecolor{MGreen}{rgb}{0,0.8,0}
\definecolor{DGreen}{rgb}{0,0.55,0}
\definecolor{Yellow}{rgb}{1,1,0}
\definecolor{Cyan}{rgb}{0,1,1}
\definecolor{Magenta}{rgb}{1,0,1}
\definecolor{Orange}{rgb}{1,.5,0}
\definecolor{Violet}{rgb}{.5,0,.5}
\definecolor{Purple}{rgb}{.75,0,.25}
\definecolor{Brown}{rgb}{.75,.5,.25}
\definecolor{Grey}{rgb}{.5,.5,.5}
\definecolor{Black}{rgb}{0,0,0}

\newenvironment{proofof}[1]{\noindent{\textbf{#1: }}}
{$\blacksquare$\vskip\belowdisplayskip}

\providecommand\abs[1]{\lvert#1\rvert}

\def\path{{\tt path}}

\newcommand{\ignore}[1]{}

\newcommand{\eps}{\varepsilon}

\newcommand{\bdm}{\begin{displaymath}}
\newcommand{\edm}{\end{displaymath}}
\newcommand{\bea}{\begin{eqnarray*}}
\newcommand{\eea}{\end{eqnarray*}}
\newcommand{\bean}{\begin{eqnarray}}
\newcommand{\eean}{\end{eqnarray}}

\newcommand{\B}{\{0, 1\}}
\newcommand{\sd}{d_{tv}}

\newtheorem{theorem}{Theorem}
\newtheorem{proposition}{Proposition}

\newtheorem{definition}{Definition}

\newtheorem{lemma}{Lemma}
\newtheorem{claim}{Claim}

\newtheorem{remark}{Remark}

\newenvironment{proof}{\noindent{\textbf{Proof: }}}
{\hfill $\blacksquare$\vskip\belowdisplayskip}

\def\np{{\rm NP}}

\def\szk{{\rm SZK}}

\def\pspace{{\rm PSPACE}}

\def\am{{\rm AM}}

\def\co{{\rm co}}

\providecommand\floor[1]{\lfloor#1\rfloor}

\providecommand\abs[1]{\lvert#1\rvert}
\newcommand{\bigabs}[1]{\bigl|#1\bigr|}

\newcommand{\tmax}{t_{\max}}

\title{The Computational Complexity of Estimating Convergence Time}
\author{Nayantara Bhatnagar \thanks{Department of Statistics, UC
    Berkeley, {\tt nayantara@huji.ac.il}, Supported by DOD ONR grant
  N0014-07-1-05-06 and DMS 0528488. Part of this work was done while
  visiting the Department Of Mathematics at the Weizmann Institute.}
\and
Andrej Bogdanov
\thanks{Department of Computer Science and Engineering, The
  Chinese University of Hong Kong, {\tt andrejb@cse.cuhk.edu.hk}, Supported by
  the National Basic Research Program of China Grant (973 Project No. 2007
  CB 807900) and a CUHK Direct Faculty Grant.}
\and
Elchanan Mossel
\thanks{Department of Statistics and Department of Computer
  Science , UC Berkeley, and Faculty of Mathematics and Computer Science, Weizmann
  Institute, {\tt mossel@stat.berkeley.edu},
  Supported by  DMS 0548249 (CAREER) award, by ISF grant 1300/08 and by EU grant
PIRG04-GA-2008-239317}}

\begin{document}
\maketitle

\begin{abstract}

An important problem in the implementation of Markov Chain Monte Carlo
algorithms is to determine the convergence time, or
the number of iterations before the chain is close to stationarity.
For many Markov chains used in practice this time is not known. Even
in cases where the convergence time is known
to be polynomial, the theoretical bounds are often too crude to be
practical. Thus, practitioners like to carry out some form of
statistical analysis in order to assess convergence.
This has led to the development of a number of methods known as convergence
diagnostics which attempt to diagnose whether the Markov chain is far
from stationarity.
  We study the problem of testing convergence in the
following settings and prove that the problem is hard in a
computational sense:
\begin{itemize}

\item Given a Markov chain that mixes rapidly, it is hard for
Statistical Zero Knowledge ($\szk$-hard) to distinguish whether starting
from a given state, the chain is close to stationarity by time $t$ or
far from stationarity at time $ct$ for a constant $c$. We show the
problem is in $\am$ intersect $\co\am$.

\item Given a Markov chain that mixes rapidly it is
  $\co\np$-hard to distinguish whether it is close to stationarity by
  time $t$ or far from stationarity at time $ct$ for a constant $c$. The
  problem is in $\co\am$.

\item It is $\pspace$-complete to distinguish whether the Markov chain
  is close to stationarity by time $t$ or far from being mixed at
  time $ct$ for $c \ge 1$.
\end{itemize}
\end{abstract}

\section{Introduction}
Markov Chain Monte Carlo (MCMC) simulations are an important tool for
sampling from
high dimensional distributions in Bayesian inference, computational
physics and biology and in applications such as image
processing. An important problem that arises in the implementation is
that if bounds on the convergence time are not known or impractical
for simulation then one would like a method
for determining if the chain is still far from converged.

A number of techniques are known to theoretically bound the rate of
convergence time as measured by the {\em mixing time} of a Markov
chain, see  e.g.~\cite{AlFi,Jer-book,LPW-book}. These
have been applied with to
problems such as volume estimation \cite{LV03}, Monte Carlo integration of
log-concave functions \cite{LV06}, approximate counting
of matchings \cite{JSV} and estimation of partition functions from
physics \cite{JS}.
However, in most practical applications of MCMC, there are no
effective bounds on the convergence time so for example it may not be
known if a chain on $2^{100}$ states mixes in time $1000$ or $2^{50}$.
Even in the cases where rapid mixing is known, the bounds are often
impractical since they are not tight especially since applications
usually require multiple independent samples. 

As a result, practitioners have focused on the development of a large
variety of statistical methods, called convergence diagnostics
which try to determine whether the Markov chain is far from
stationarity (see e.g. surveys by
\cite{IN-book,BR,CC,CL-book,GRS-book,RC-book}). A majority of
practitioners of the MCMC method run multiple diagnostics  
to test if the chains have converged. The two most popularly used public
domain diagnostic software packages are CODA and BOA \cite{PBCV, Boa}.
The idea behind
many of the methods is to
use the samples from the empirical distribution obtained when running
one or multiple copies of the chain, possibly from multiple
starting states to compute various
functionals and identify non-convergence.

While diagnostics are commonly used for MCMC, it has been repeatedly
observed that they cannot guarantee  
convergence, see e.g. ~\cite{CC,BR,AGT}. 


Here we formalize convergence to stationarity detection as an
algorithmic problem and study its complexity in terms of the size of
the description of the Markov chain, denoted by $n$. 
Our main contribution is showing that even in cases
where the mixing time of the chain is known to be bounded by $n^C$ for
some large $C$,  
the problem of distinguishing whether a Markov chain is close to or
far from stationarity at time $n^c$ for $c$ much smaller than $C$ is
``computationally hard". In other words under standard assumptions in
computational complexity the problem of distinguishing whether the
chain is close to or far from stationarity cannot
be solved in time $n^D$ for any constant $D$.  

The strength of our results is in their generality as they apply to
{\em all possible} diagnostics and in the weakness of the assumption -
in particular in assuming that the mixing time of the chain is not too long and
that the diagnostic is also given the initial state of the chain.  

From the point of view of theoretical computer science, our results
highlight the role of Statistical Zero Knowledge, $\am$, $\co\am$ and $\co\np$ in 
the computational study of MCMC.





\section{Results}
We begin by defining the mixing time which measures the rate of
convergence to the stationary distribution. Recall that the {\em variation
distance} (or statistical distance) between two probability
distributions $\mu$ and $\nu$ on state space $\Omega$ is given by
$d_{tv}(\mu,\nu) = \frac12 \sum_{\omega \in \Omega} \left|\mu(\omega) -
  \nu(\omega)\right|$.

\begin{definition}[Mixing time]
Let $M$ be a Markov chain with state space $\Omega$, transition matrix
  $P$ and a unique stationary distribution $\pi$.
The following measure of distance to stationarity will be convenient to define:
$$d(t) := \max_{x,y \in \Omega}d_{tv}(P^{t}(x,
\cdot),P^t(y,\cdot)).$$

The {\em $\varepsilon$-mixing
  time} is defined to be,
$$\tau(\varepsilon) := \min \{t: d(t) \leq \varepsilon \}.$$
We refer to $\tau(1/4)$ as the {\em mixing time}.
We also define the {\em $\varepsilon$-mixing time starting from $x$}:
$$\tau_x(\varepsilon) := \min \{t: \ d_{tv}(P^{t}(x,\cdot),\pi) \leq
\varepsilon \}.$$

We note that $\tau_x(\varepsilon) \leq \tau(\varepsilon)$ for all $x$.


\end{definition}

To formulate the problem, we think of the Markov chain as a ``rule"
for determining the next state of the chain  
given the current state and some randomness. 
\begin{definition}
We say that a circuit $C: \B^n \times \B^m \to \B^n$ {\em
  specifies} $P$ if for every pair of states $x, y \in \Omega$,
$\Pr_{r \sim \B^m}[C(x, r) = y] = P(x, y)$.
\end{definition}
In this formalization, $x$ is the ``current state", $r$ is the
``randomness", $y$ is the ``next state" and $C$ is the ``rule". 
Next we formalize the notion of ``Testing convergence''. 
We imagine the practitioner
has a time $t$ in mind that she would like to run the Markov chain algorithm
for. She would like to use the diagnostic to
determine whether at time $t$: 
\begin{itemize}
\item
The chain is (say) within $1/4$ variation distance of stationarity
\item
or at least at distance $1/4$ away from it.
\end{itemize}
Requiring the diagnostic to determine the total variation at time $t$ exactly is
not needed in many situations.  

Many practitioners will be happy with a diagnostic which will 
\begin{itemize}
\item
Declare the chain has mixed if it is within $1/8$ variation distance
of stationarity at time $t$.
\item
Declare it did not mix if it is at least at distance $1/2$ away from
it at time $t$.  
\end{itemize}

An even weaker requirement for the diagnostic is to:
\begin{itemize}
\item
Declare the chain has mixed if it is within $1/8$ variation distance
of stationarity at time $t$. 
\item
Declare it did not mix if it is at least at distance $1/2$ away from
it at time $c t$, where $c \ge 1$.  
\end{itemize}
Thus in the last formulation, the practitioner is satisfied with an
approximate output of the 
diagnostics both in terms the time and  
in terms of the total variation distance. This is the problem we will
study. In fact, we will make the requirement from the diagnostic even easier  
by providing it with a (correct) bound on the actual mixing time of
the chain. This bound will be denoted by $t_{\max}$. 

In realistic settings it is natural to measure the running time of the
diagnostics in relation to the running time  
of the chain itself as well as to the size of the chain. In particular
it is natural to consider diagnostics that would run for 
time that is polynomial in $t$ and $t_{\max}$. The standard way to
formalize such a requirement is to insist that the inputs $t,t_{\max}$
to the  
diagnostic algorithm to be given in unary form (note that if
$t,t_{\max}$ were specified as 
binary numbers, an efficient algorithm would be required to run in time
{\em poly-logarithmic} in these parameters, a much stronger
requirement). We continue with description of the different
diagnostic problems and the statement of the hardness results.


\subsection{Given Starting Point} 
The discussion above motivates the definition of
the first problem below. Assume that we had a diagnostic algorithm. As
input, it would take the the tuple $(C,x,1^t,1^{t_{\max}})$, i.e., a
description of the circuit which describes the moves of the Markov
chain, an initial starting state for the chain, and the times $t$ and
$t_{\max}$, which are specified as unary numbers. 
The following theorems show that a diagnostic algorithm
as described above is unlikely to exist under standard
complexity-theoretic assumptions.
We consider two versions of the convergence testing problem, one where the
starting state of the Markov chain is specified ({\sc GapPoly
  TestConvergenceWithStart$_c$}) and the other where it is arbitrary ({\sc
  GapPolyTestConvergence$_c$}):\\

{\bf Problem:} {\sc GapPolyTestConvergenceWithStart$_{c,\delta}$}
({\sc GPTCS$_{c,\delta}$}).\\
{\bf Input:} $(C,x,1^t,1^{\tmax})$, where $C$ is a circuit specifying a Markov
chain $P$ on state space $\Omega \subseteq \{0,1\}^n$, $x \in \Omega$  and
$t,t_{\max} \in \mathbb N$.\\
{\bf Promise:} The Markov chain $P$ is ergodic and $\tau(1/4) \leq \tmax$.\\
{\bf YES instances:} $\tau_x(1/4-\delta) < t.$\\
{\bf NO instances:} $\tau_x(1/4+\delta) > ct.$\\

Informally the input to this problem is the MC rule $C$, a starting
state $x$, and times $t,t_{\max}$.  
It is promised that the chain mixes by time $\tmax$. The expectation
from the diagnostic is to:  
\begin{itemize}
\item
Declare the chain has mixed if it is within $1/4-\delta$ variation
distance of stationarity at time $t$. 
\item
Declare it did not mix if it is at least at distance $1/4+\delta$ away
from it at time $c t$, where $c > 1$. 
\end{itemize} 
Note again that the diagnostic is given room for error both in terms
of the total variation distance and in terms  
of the time. 

The following theorem refers to the complexity class $\szk$, which is
the class of all promise problems that have statistical zero-knowledge
proofs with completeness $2/3$ and soundness $1/3$. It is believed
that these problems cannot be solved in polynomial time.  
 
\begin{theorem}\label{thm:szk-complete}Let $c \ge 1$.
\begin{itemize}
\item For $0 < \delta \le 1/4$, {\sc
    GPTCS$_{c,\delta}$} is in $\am \cap \co\am$.
\item For $\frac{\sqrt{3} - 1.5}{2}  = .116025..< \delta \le 1/4$, {\sc
    GPTCS$_{c,\delta}$} is in $\szk$.
\item Let $0 \leq \delta < 1/4$. For $$c<\frac{\tmax}{4t}
\ln \left(\frac{2}{1+4\delta}\right),$$ {\sc GPTCS$_{c,\delta}$} is $\szk$-hard.
\end{itemize}
\end{theorem}

The most interesting part of the theorem is the last part which
informally says that the problem  
{\sc GPTCS$_{c,\delta}$} is $\szk$-hard. In other words, solving it in
polynomial time will result in solving all the problems 
in $\szk$ in polynomial time. The second part of the theorem states
that for some values of $\delta$ this is the ``exact" level of
hardness.  
The first part of the theorem states that without restrictions on
$\delta$ the problem 
belongs to the class $\am \cap \co\am$ (which contains the
class $\szk$). The classes $\am$ and $\co\am$ respectively contain the
classes $\np$ and $\co\np$ and it is believed that they are equal to
them, but this is as yet unproven. 


The restriction on the constant $\delta$ in the second part of the result comes
from the fact that the proof
is by reduction to the $\szk$-complete problem {\sc Statistical
  Distance} ({\sc SD}, see Section \ref{sec:SD} for precise
definitions). Holenstein and Renner give evidence in \cite{HR} that
{\sc SD} is in $\szk$ only when there
is a lower bound on
the gap between the completeness and soundness. We show that the
restriction in Theorem 
\ref{thm:szk-complete}
necessary since otherwise it would be possible to put {\sc SD} in
$\szk$ for a smaller value of the completeness-soundness gap.

On the other hand, we can show a slightly weaker result and put  {\sc
    GPTCS$_{c,\delta}$} into $\am \cap \co\am$ without any restrictions on
$\delta$. To show this, we first prove that {\sc SD} is in $\am \cap
\co\am$ when no restriction is put on the
gap between the completeness and soundness. This result may be
interesting in its own right as it involves showing protocols for {\sc
  Statistical Distance} that are new, to our knowledge.

\subsection{Arbitrary Starting Point}
So far we have discussed mixing from a given starting point. 
A desired property of a Markov chain is fast mixing from an arbitrary
starting point. Intuitively, this problem is harder than the previous one 
since it involves all starting points. This is consistent with our result below 
where we obtain a stronger
hardness. \\

{\bf Problem:} {\sc GapPolyTestConvergence$_{c,\delta}$} ({\sc
  GPTC$_{c,\delta}$}).\\
{\bf Input:} $(C,x,1^t,1^{\tmax})$, where $C$ is a circuit specifying a Markov
chain $P$ on state space $\Omega \subseteq \{0,1\}^n$, $x \in \Omega$  and
$t,\tmax \in \mathbb N$.\\
{\bf Promise:} The Markov chain $P$ is ergodic and
$\tau(1/4) \leq \tmax$.\\
{\bf YES instances:} $\tau(1/4-\delta) < t$.\\
{\bf NO instances:} $\tau(1/4+\delta) > ct$.

Note that the only difference between this and the previous problem is
that the total variation  
distance is measured from the worst starting point instead of from a
given starting point.

\begin{theorem}\label{thm:conp-hard}
Let $c \geq 1$.
\begin{itemize}
\item For $0<\delta \le 1/4$, {\sc GPTC$_{c,\delta}$} $\in \co\am$.
\item Let $0 \leq \delta <1/4$. For 
\[
c< \frac{3/4-\delta}{2} \sqrt{\tmax/t^2 n^3}
\]
 it is $\co\np$-hard to decide {\sc
    GPTC$_{c,\delta}$}.
\end{itemize}
\end{theorem}

Again the second part of the theorem is the more interesting part. It
shows that the diagnostic problem is $\co\np$ hard so it is very
unlikely to be solved in polynomial time. This hardness is stronger
than $\szk$-hardness because $\szk$ is unlikely to contain $\co\np$-hard
problems. If it did, this would imply that $\np=\co\np$ since $\szk
\subseteq \am$ and it is believed that $\am=\np$. The
first part of the theorem shows that the problem is always in $\co\am$.

\subsection{Arbitrary mixing times}
Finally we remove the restriction that the running time of the
algorithm should be polynomial in the times $t,t_{\max}$.  
This corresponds to situations where the mixing time of the chain may be
exponentially large in the size of the rule defining the
chain. 
This rules out many situations of practical interest. However it is
relevant in scenarios where analysis of the mixing time  
is of {\em theoretical} interest. For example there is an extensive
research in theoretical physics on the rate of convergence of Gibbs
samplers on spin glasses even in cases where the convergence rate is
very slow (see~\cite{Der:87} and follow up work). 
In such setups it is natural to define the problem as follows:\\


{\bf Problem} {\sc GapTestConvergence$_{c,\delta}$} ({\sc GTC$_{c,\delta}$}).\\
{\bf Input:} $(C,x,t)$, where $C$ is a circuit specifying a Markov
chain $P$ on state space $\Omega \subseteq \{0,1\}^n$, $x \in \Omega$  and
$t \in \mathbb N$.\\ 
{\bf Promise:} The Markov chain $P$ is ergodic.\\
{\bf YES instances:} $\tau(1/4 - \delta) < t$.\\
{\bf NO instances:} $\tau(1/4+\delta) > ct$.\\

Note that the main difference is that in this problem the time $t$ is given in
binary representation. Thus, informally in this case   
the efficiency is measured with respect to the logarithm of
$t$. Additionally note that the mixing time of the chain itself does
not put any restrictions on the diagnostic. We then prove the
following result:

\begin{theorem}\label{thm:pspace-complete}
Let $1 \le c \le \exp(n^{O(1)})$.
\begin{itemize}
\item For $\exp(-n^{O(1)})< \delta \leq 1/4$
it is in $\pspace$ to decide {\sc GTC$_{c,\delta}$}.
\item Let $ 0 \leq \delta <
1/4$, then, it is  $\pspace$-hard to decide {\sc GTC$_{c,\delta}$}.
\end{itemize}
\end{theorem}

It is known that $\pspace$ hard problems are at least as hard as all
the problem in polynomial time $\co\np$, $\np$ and all other problems  
in the polynomial hierarchy.  


\section{Protocols for statistical distance}\label{sec:SD}

Given a circuit $C\colon \B^n \to \B^n$, the probability distribution
$p_C$ associated to $C$ assigns probability $p(\omega) =
\abs{C^{-1}(\omega)}/2^n$ to every $\omega \in \B^n$. We will be
interested in estimating the statistical distance between the
distributions associated to a pair of circuits $C, C'\colon \B^n \to
\B^n$. Denote those distributions by $p$ and $p'$,
respectively.

For a pair of constants $0 \leq \mathbf s < \mathbf c \leq 1$, {\sc
  SD}$_{\mathbf c,\mathbf s}$
is defined to be the following promise problem. The inputs are pairs
of circuits $C,
C'\colon \B^n \to \B^n$, the YES instances satisfy $\sd(p, p') \geq
\mathbf c$, and the NO instances satisfy $\sd(p, p') < \mathbf s$.

Sahai and Vadhan~\cite{SaVa} show that for every pair of constants
$\mathbf c,\mathbf s$ the problem
{\sc SD$_{\mathbf c,\mathbf s}$} is $\szk$-hard. They also show that
when $\mathbf c^2 > \mathbf s$, $SD_{\mathbf c, \mathbf s}$ is in
$\szk$. Our theorem yields a weaker conclusion, but covers a wider
spectrum of parameters.

\begin{theorem}\label{thm:am-coam}
For any pair of constants $0 \leq \mathbf s < \mathbf c \leq 1$,
$SD_{\mathbf c, \mathbf  s}$ is in
$\am \cap \co\am$.
\end{theorem}

\subsection{An $\am$ protocol}
The following interactive protocol $P$ for $SD_{\mathbf c, \mathbf s}$
essentially appears in \cite{SaVa} but we rewrite it here for the precise
parameters we need:

\begin{itemize}
\item[\bf V:] Flip a fair coin. If heads, generate a random sample
  from $C$. If tails, generate a random sample from $C'$. Send the
  sample $x$ to the prover.
\item[\bf P:] Say if $x$ came from $C$ or from $C'$.
\item[\bf V:] If prover is correct accept, otherwise reject.
\end{itemize}

\begin{claim}\label{cl:AM-SD}
Protocol $P$ is an interactive proof for $SD_{\mathbf c, \mathbf s}$
with completeness
$1/2 +\mathbf c$ and soundness $1/2 + \mathbf s$.
\end{claim}
\begin{proof}
We prove soundness first. Let $T$ be the set of $x$s which the prover
claims came from $C$. The accepting probability is
\[ \sum\nolimits_{x \in T} \frac{p(x)}{2} + \sum\nolimits_{x \not\in
  T} \frac{p'(x)}{2}
   = \frac{1}{2} \bigl(\sum\nolimits_{x \in T} p(x) + \sum\nolimits_{x
     \not\in T} p'(x)\bigr). \]
No matter what $T$ is, we have that $$\frac{1}{2}(\sum_{x \in T}
p(x)+\sum_{x \not\in T} p'(x)) = \frac{1}{2}(1- \sum_{x \not\in T}
p(x)+\sum_{x \not\in T} p'(x)) \leq 1/2 + \sd(p, p'),$$ and so the
accepting probability is at most $1/2+\mathbf s$.

To prove completeness, notice that the above inequality is tight when
$T$ equals the set of those $x$ such that $p(x) > p'(x)$. So when the
prover uses this strategy (say $C$ if $p(x) > p'(x)$ and $C'$
otherwise), the accepting probability becomes exactly $1/2 +
\sd(p,p')\geq  1/2+\mathbf c$.
\end{proof}

\subsection{A $\co\am$ protocol}
Showing that  $SD_{\mathbf c, \mathbf s}$ is in $\co\am$ is a bit more
involved. Such
a protocol wants to accept when the statistical distance between $p$
and $p'$ is small, and reject when the statistical distance is
large. To develop some intuition, let us first attempt to distinguish
the cases when $p$ and $p'$ are the same distribution (i.e. $\mathbf s = 0$)
and the case when they are at some distance from one another (say $\mathbf c =
1/2$).

Let's forget for a moment that the verifier has to run in polynomial time. Suppose the verifier could get hold of the values
\[
N(t) = |\bigabs{\{\omega\colon \text{$\abs{C^{-1}(\omega)} \geq t$ and $\abs{C'^{-1}(\omega)} \geq t$}\}}|
\]
for every $t$ (which could potentially range between $0$ and $2^n$). Then it can compute the desired statistical distance via the following identity which will be proven later:
\begin{equation}
\label{eqn:sdother}
\sum_{t = 1}^{2^n} t \cdot (N(t) - N(t+1)) = (1 - \sd(p, p')) \cdot 2^n.
\end{equation}
If we want the verifier to run in polynomial time, there are two issues with this strategy: First, the verifier does not have time to compute the values $N(t)$ and second, the verifier cannot evaluate the exponentially long summation in (\ref{eqn:sdother}). If we only want to compute the statistical distance approximately, the second issue can be resolved by quantization: Instead of computing the sum on the left for all the values of $t$, the verifier chooses a small number of representative values and estimates the sum approximately. For the first issue, the verifier will rely on the prover to provide (approximate) values for $N(t)$. While the verifier cannot make sure that the values provided by a (cheating) prover will be exact, she will be able to ensure that the prover never grossly over-estimates the sum on the left by running a variant of the Goldwasser-Sipser protocol which we describe below. Since the sum on the left is proportional to one minus the statistical distance, it will follow that no matter what the prover's strategy is, he cannot force the verifier to significantly underestimate the statistical distance without being detected.

We now give the details of this protocol, starting with a proof of~(\ref{eqn:sdother}).

\begin{proofof}{Proof of identity (\ref{eqn:sdother})}
Let $f(\omega) = \min\{\abs{C^{-1}(\omega)}, \abs{C'^{-1}(\omega)}\}$. Then
\[ \sum_{\omega \in \B^n} f(\omega) = \sum_{t = 1}^{2^n} t \cdot \bigabs{\{\omega\colon f(\omega) = t\}}
   = \sum_{t = 1}^{2^n} t \cdot \bigl(\abs{\{\omega\colon f(\omega) \geq t\}} - \abs{\{\omega\colon f(\omega) \geq t + 1\}}\bigr). \]
The right-hand side of this expression is exactly equal to the left-hand side of (\ref{eqn:sdother}). For the left-hand size, using the formula $\min\{a, b\} = (a + b)/2 - \abs{a - b}/2$ (where $a, b \geq 0$) we have
\[
\sum_{\omega \in \B^n} f(\omega)
= \frac12 \sum_{\omega \in \B^n} \bigl(\abs{C^{-1}(\omega)} + \abs{C'^{-1}(\omega)}\bigr) - \frac12\sum_{\omega \in \B^n} \bigabs{\abs{C^{-1}(\omega)} - \abs{C'^{-1}(\omega)}}
= 2^n - \sd(p, p') \cdot 2^n
\]
which equals the right-hand side of (\ref{eqn:sdother}).
\end{proofof}

\paragraph{A lower bound protocol for $N(t)$}
We now show that a variant of the Goldwasser-Sipser lower bound protocol can be used to certify lower bounds on the quantities $N(t)$. More precisely, we design an $\am$ protocol for the following problem:

\medskip
\noindent{\bf Input:} A pair of circuits $C, C'\colon \B^n \to \B^n$, a number $1 \leq t \leq 2^n$, a target number $0 \leq \tilde{N} \leq 2^n$, and a fraction $0 < \delta \leq 1$ (represented in unary). \\
\noindent{\bf Yes instances:} $(C, C', t, \tilde{N}, \delta)$ such that $N(t) \geq \tilde{N}$ \\
\noindent{\bf No instances:} $(C, C', t, \tilde{N}, \delta)$ such that $N((1 - \delta)t) < (1 - \delta)\tilde{N}$.
\medskip

Here is a protocol for this problem. Here, $\delta_1, \delta_2$ are the largest values below $\delta$ that make the logarithms below integers. In the analysis, for simplicity we will assume that $\delta_1 = \delta_2 = \delta$.
\begin{itemize}
\item[{\bf V:}] Set $a = \log(\delta_1^2 \tilde{N}/54)$. Send a random hash function $g\colon \B^n \to \B^a$.
\item[{\bf P:}] Let $c = \floor{(1 - \delta_1/2)(54/\delta_1^2)}$. Send a set of values $\{\omega_1, \dots, \omega_c\}$.
\item[{\bf V:}] Set $b = \log(\delta_2^4 t/5000)$. Send a random hash function $h\colon \B^n \to \B^b$.
\item[{\bf P:}] Let $d = \floor{(1 - \delta_2/2)(5000/\delta^4)}$. For each $1 \leq i \leq c$, send sets $\{r_{i1},\dots,r_{id}\}$ and $\{r'_{i1},\dots,r'_{id}\}$.
\item[{\bf V:}] If $g(\omega_i) = 0$ for all $i$ and $h(r_{ij}) = h(r'_{ij}) = 0$ and $C(r_{ij}) = C'(r'_{ij}) = \omega_i$ for all pairs $(i, j)$, accept, otherwise reject.
\end{itemize}

We first prove completeness: If $(C, C', t, \tilde{N}, \delta)$ is a yes instance, the protocol accepts with probability at least $2/3$. Let
\[ S = \{\omega\colon \text{$\abs{C^{-1}(\omega)} \geq t$ and $\abs{C'^{-1}(\omega)} \geq t$}\}. \]
The expected number of $\omega \in S$ with $g(\omega) = 0$ is at least $(54/\delta^2)\cdot(N(t)/\tilde{N})$. If $N(t) \geq \tilde{N}$, by Chebyshev's inequality, the probability over $g$ of getting fewer than $c = (1-\delta)(54/\delta^2)$ such $\omega_i$s is at most $1/6$. Assuming all these $\omega_i$s exist, let's fix one of them. We now look at the set $T_i = \{r\colon C(r) = \omega_i\}$. Since $\omega_i \in S$, $T$ has size at least $t$, so the expected number $r \in T_i$ such that $h(r) = 0$ is at least $5000/\delta^4$. By Chebyshev's inequality, the probability of getting fewer than $d$ such $r_{ij}$s is at most $\delta^2/1248$. This bound holds for every $i$ and also for the sets $T'_i = \{r\colon C'(r) = \omega_i\}$. Taking a union bound over all $2c$ such sets we get that with probability at least $5/6$ over the choice of $h$, a sufficient number of $r_{ij}$s and $r'_{ij}$s exist for all values of $i$, so the verifier accepts.

We now prove soundness: If $(C, C', t, \tilde{N}, \delta)$ is a no instance, the protocol accepts with probability at most $2/3$. Now let
\[ S = \{\omega\colon \text{$\abs{C^{-1}(\omega)} \geq (1 - \delta)t$ and $\abs{C'^{-1}(\omega)} \geq (1 - \delta)t$}\}. \]
The expected number of $\omega \in S$ with $g(\omega) = 0$ is then at most $(1-\delta)(54/\delta^2)$. In this case, $c$ is at least equal to $(1 + \delta/3)$ times this expected value. By Chebyshev's inequality, the probability that there exist $c$ such $\omega_i$s is then less than $1/6$. If not, then the prover is forced to send at least one $\omega_i$ such that either $g(\omega_i) \neq 0$ or $\omega_i \not\in S$. In the first case, the verifier rejects. In the second case, we let
\[ T_i = \{r\colon C(r) = \omega_i\} \qquad \text{and} \qquad T'_i = \{r\colon C'(r) = \omega_i\} \]
so either $\abs{T_i} < (1-\delta)t$ or $\abs{T'_i} < (1-\delta)t$. Without loss of generality, let us assume the first case. Then the expected number of $r \in T$ such that $h(r) = 0$ is at most $(1-\delta)(5000/\delta^4)$. We apply Chebyshev's inequality again to conclude that with probability at least $5/6$, the prover is then forced to send some $r_{ij}$ such that either $h(r_{ij}) \neq 0$ or $C(r_{ij}) \neq \omega_i$. Thus the verifier accepts with probability at most $1/6 + 1/6 \leq 1/3$.

Repeating this protocol in parallel sufficiently many times, we have the following consequence, which we will use below:

\begin{claim}
\label{claim:lower}
There is an $\am$ lower bound protocol for $N(t)$ with completeness $1 - \delta/20n$ and soundness $\delta/20n$.
\end{claim}

\paragraph{A $\co\am$ protocol for statistical distance}
We now give the $\co\am$ protocol for statistical distance. We begin
with the observation that it is sufficient to handle the following
special case of the problem:

\medskip
\noindent{\bf Input:} A pair of circuits $C, C'\colon \B^n \to \B$ and
a fraction $0 < \delta \leq 1/3$ (represented in unary). \\
\noindent{\bf Yes instances:} $(C, C', \delta)$ such that $\sd(p, p')
\leq \delta$ \\
\noindent{\bf No instances:} $(C, C', \delta)$ such that $\sd(p, p') > 3\delta$.
\medskip

We can reduce $SD_{c,s}$ for any pair of constants $0 \leq s < c \leq
1$ to the above problem via the XOR lemma of Sahai and
Vadhan~\cite{SaVa}, which reduces $SD_{c, s}$ to $SD_{c^k, s^k}$ for an
arbitrary constant $k$. When $k$ is chosen so that $(c/s)^k > 3$, the
resulting instance can be handled by our protocol.

We now give the protocol for statistical distance:
\begin{itemize}
\item[{\bf P:}] Send claims $\tilde{N}_i$ for the values $N_i = N((1 - \delta)^{-i})$, $0 \leq i \leq en/\delta$.
\item[{\bf P, V:}] Run the $\am$ lower bound protocol for $N_i$ on inputs $(C, C', (1 - \delta)^{-i}, \tilde{N}_i, \delta)$ for every $1 \leq i \leq en/\delta$. If all of them pass accept, otherwise reject.
\item[{\bf V:}] Accept if $\sum_{i=0}^{en/\delta} (\tilde{N}_i  -\tilde{N}_{i+1})(1 - \delta)^{-i} \geq (1-\delta)^2 \cdot 2^n$.
\end{itemize}

The soundness and completeness rely on the following approximation, which is a quantized version of (\ref{eqn:sdother}):
\begin{equation}
\label{eqn:sdapprox}
\sum_{i=0}^{en/\delta} (N_i - N_{i+1})(1 - \delta)^{-i} \leq (1 - \sd(p, p'))2^n \leq \sum_{i=0}^{en/\delta} (N_i - N_{i+1})(1 - \delta)^{-(i+1)}.
\end{equation}
This is proved in a similar way as (\ref{eqn:sdother}). For every $i$, we have the sandwiching inequality
\[ (N_i - N_{i+1})(1 - \delta)^{-i} \leq \sum\nolimits_{\omega\colon f(\omega) \in [(1-\delta)^{-i}, (1 - \delta)^{-(i+1)})} f(\omega)
   \leq (N_i - N_{i+1})(1 - \delta)^{-(i+1)}, \]
which yields (\ref{eqn:sdapprox}), after summing over all $i$ from $0$ to $en/\delta$.

To prove completeness, consider an honest prover which claims $\tilde{N}_i = N_i$ for all $i$. By Claim~\ref{claim:lower} and a union bound, with probability at least $2/3$ none of the lower bounds protocols for $N_i$ reject. In this case, using (\ref{eqn:sdapprox}), we get
\[ \sum_{i=0}^{en/\delta} (\tilde{N}_i  -\tilde{N}_{i+1})(1 - \delta)^{-i} \geq (1 - \delta)(1 - \sd(p, p')) \cdot 2^n \]
establishing completeness. To prove soundness, assume now that the verifier accepts with probability at least $1/3$. By the soundness of the lower bound protocols for $N_i$ (Claim~\ref{claim:lower}) and a union bound, there must exist at least one setting of the randomness of the verifier for which $N_{i-1} \geq (1 - \delta)\tilde{N}_i$ for all $i$ (where $N_{-1} = N_0$) and the verifier accepts. Now (using the fact that the last value of $N_i$ is zero):
\begin{align*}
\sum_{i=-1}^{en/\delta} (N_i - N_{i+1}) (1 - \delta)^{-i}
&= \frac{N_{-1}}{1-\delta} + \sum_{i=0}^{en/\delta-1} N_i\bigl((1 - \delta)^{-(i+1)} - (1 - \delta)^{-i})\bigr) \\
&\geq \tilde{N}_0 + \sum_{i=0}^{en/\delta-1} (1 - \delta)\tilde{N}_{i+1}\bigl((1 - \delta)^{-(i+1)} - (1 - \delta)^{-i})\bigr) \\
&= \delta\tilde{N}_0 + (1 - \delta) \cdot \sum_{i=0}^{en/\delta} (\tilde{N}_i  -\tilde{N}_{i+1})(1 - \delta)^{-i} \\
&\geq (1 - \delta)^3 \cdot 2^n
\end{align*}
so from (\ref{eqn:sdapprox}) we get that $1 - \sd(p, p') \geq (1 - \delta)^3$, so $\sd(p, p') \leq 1 - (1 - \delta)^3 \leq 3\delta$.

\section{Diagnosing Convergence for Polynomially Mixing
  Chains}\label{sec:poly-mixing}

The results of this section imply that even if the mixing time is
restricted to being polynomial the diagnostic problem remains
hard. The two cases we consider are the worst case start mixing time
and the mixing time from a given starting state. Both hardness results
are by reduction from a complete problem in the respective classes. We
first prove Theorem \ref{thm:szk-complete}.

\begin{lemma}\label{lem:in-szk} The problem {\sc GPTCS$_{c,\delta}$} is in
  $\szk$ for all $c \geq 1$ and $\frac{\sqrt{3} - 1.5}{2}  = .116025...<
  \delta \le 1/4$.
\end{lemma}
\begin{proof}
The proof is by reduction to {\sc SD$_{\mathbf c,\mathbf s}$} where
$\mathbf c$ and $\mathbf s$ are
chosen as follows.
Choose $k$ large enough such that
\begin{eqnarray}\label{eq:c-s-condition}
\left(\frac{1}{4} +\delta - \frac{1}{k}\right)^2 > \frac{1}{4}
-\delta +\frac{1}{k}.
\end{eqnarray}
Let $$\mathbf s = \frac{1}{4} - \delta + \frac{1}{k}$$ and $$\mathbf c = \frac{1}{4} + \delta
-\frac{1}{k}.$$

Suppose we are given an instance of {\sc GPTCS}$_{c,\delta}$ with
input $(C,x,1^t,1^{t_{\max}})$. Let $\tau = \tau(1/k)$ be the time to come within
$1/k$ in variation distance of the stationary distribution. Let
$C$ output
the distribution $P^t(x,\cdot)$ over $\Omega$. Let $C'$ output the
distribution $P^{\tau}(x,\cdot)$ over $\Omega$.
In the YES case,
$$|P^t(x,\cdot) - P^{\tau}(x,\cdot)| \leq \frac{1}{4} -\delta + \frac{1}{k}$$
while in the NO case,
$$|P^{ct}(x,\cdot) - P^{\tau}(x,\cdot)| > \frac{1}{4} +
\delta -
\frac{1}{k}.$$
Since $c \geq 1$, this implies that
$$|P^t(x,\cdot) - P^{\tau}(x,\cdot)| > \frac{1}{4} +
\delta -
\frac{1}{k}.$$
By (\ref{eq:c-s-condition}), the constructed instance of {\sc
  SD$_{\mathbf c,\mathbf s}$} is in $\szk$ and the lemma follows.
\end{proof}

\begin{lemma}\label{lem:in-am-coam} The problem {\sc GPTCS$_{c,\delta}$} is in
  $\am \cap \co\am$ for all $c\geq 1$ and $0 < \delta \le 1/4$.
\end{lemma}

This part of the result follows directly from Theorem
\ref{thm:am-coam} by reducing  {\sc GPTCS$_{c,\delta}$} to {\sc
  SD}$_{\mathbf c,\mathbf s}$ as above, without the restriction on the
gap between $\mathbf c$ and $\mathbf s$.
We can show that the gap for $\delta$ in
Lemma \ref{lem:in-szk} is required for membership in $\szk$.
Sahai and
Vadhan~\cite{SaVa} show that when $\mathbf c^2 > \mathbf s$, {\sc
  SD$_{\mathbf c,\mathbf s}$} is in $\szk$. Holenstein and Renner
\cite{HR} show that this condition on the gap between $\mathbf c$ and
$\mathbf s$ is in fact essential for membership in $\szk$.

\begin{proposition}\label{prop:small-delta-hard} There exist $\mathbf
  c,\mathbf s$
  satisfying $\mathbf c^2<\mathbf s<\mathbf c$ and $c$ such that if
  there is an $\szk$
protocol for an instance of {\sc GPTCS$_{c,\delta}$} with a
sufficiently small $\delta$, then there is an $\szk$ protocol for
{\sc SD$_{\mathbf c,\mathbf s}$}.
\end{proposition}

\begin{proof} The proof is by reduction from {\sc SD$_{\mathbf
      c,\mathbf s}$} to {\sc
  GPTCS$_{c,\delta}$}. Let $(C,C')$ be an instance
of {\sc SD$_{\mathbf c,\mathbf s}$} where $C$ and $C'$ are
  circuits which output
distributions $\mu_1$ and $\mu_2$
over $\{0,1\}^n$. Construct
the Markov chain $P$, whose state space is $[m] \times \{0,1\}^{n}$
where $m=p(n)$ is a polynomial in $n$.
The transitions of the chain are defined as follows. Let the current
state be $(X_t,Y_t)$ where $X_t \in [m]$ and $Y_t \in \{0,1\}^{n}$.
\begin{itemize}
\item If $X_t = 1$, choose $Y_{t+1}$ according to $\mu_1$.
\item If $X_t = 2$, choose $Y_{t+1}$ according to $\mu_2$.
\item Otherwise, set $Y_{t+1} = Y_t$.
\item Choose $X_{t+1}$ uniformly at random from $[m]$.
\end{itemize}

The stationary distribution of the chain is given by $\pi(z,y) =
\frac{1}{m}(\frac12 \mu_1(y)+ \frac12 \mu_2(y))$. Take the starting state to be
$x=(1,0^n)$. In one step, the
total variation distance from stationary can be bounded as
\begin{eqnarray*}
d_{tv}(P(x, \cdot), \pi) = \frac{1}{2} d_{tv}(\mu_1,\mu_2)
\end{eqnarray*}

For $t>1$, we have
\begin{align}
P^t(x,\cdot) = U_{[m]} \times
  \left(1-\left(\frac{m-2}{m}\right)^{t-1}\right)(\frac 12 \mu_1+
  \frac12 \mu_2) +
  \left(\frac{m-2}{m} \right)^{t-1} \mu_1
\end{align}
Hence, it can be verified that
\begin{align}\label{eq:t-step-dist}
d_{tv}(P^t(x,\cdot),\pi) = \frac{1}{2}
 \left(\frac{m-2}{m}\right)^{t-1}d_{tv}(\mu_1,\mu_2)
\end{align}

Let $0< \delta < (\sqrt{5}/2-1)/2$, $\mathbf  s =1/2 - 2\delta$
and $\mathbf  c =1/2 +
2\delta$ so that $\mathbf c^2<\mathbf s<\mathbf c$. Set $c=1$, $t=1$
and $t_{\max}=m $.

In the YES case, $d_{tv}(\mu_1,\mu_2) < \mathbf s$ and hence after one step,

\begin{align}
d_{tv}(P(x,\cdot),\pi) < \frac{1}{2} \mathbf s < \frac 14 -\delta
\end{align}

In the NO case, $d_{tv}(\mu_1,\mu_2) > \mathbf c $ and after one step,
\begin{align}
d_{tv}(P(x,\cdot),\pi) \geq \frac{1}{2} \mathbf c > \frac 14 +\delta.
\end{align}

From (\ref{eq:t-step-dist}) it can be seen that in both cases,
$\tau(1/4) \leq m=t_{\max}$.
This completes the reduction since if there is an $\szk$ protocol for
{\sc GPTCS$_{c,\delta}$} with the above parameters, then
it can be used to distinguish the YES and NO case of {\sc
  SD$_{\mathbf c,\mathbf s}$} for the above values of $\mathbf c,
\mathbf s$.
\end{proof}

We now complete the proof of Theorem \ref{thm:szk-complete}.

\begin{lemma}Let $0 \le \delta < 1/4$. For $1 \leq c<\frac{\tmax}{4t}
  \ln \left(\frac{2}{1+4\delta}\right)$,  the 
  problem {\sc GPTCS$_{c,\delta}$} is $\szk$-hard.
\end{lemma}
\begin{proof}
The proof uses the same reduction as in Proposition
\ref{prop:small-delta-hard} from {\sc
  SD$_{\mathbf c,\mathbf s}$}. We recall that
\begin{align}\nonumber
d_{tv}(P^t(x,\cdot),\pi) = \frac{1}{2}
 \left(\frac{m-2}{m}\right)^{t-1}d_{tv}(\mu_1,\mu_2)
\end{align}

Choose $m \geq 3$. Set $\mathbf  s=1/4 -\delta$ and $\mathbf
c=1$. Note that $\mathbf c^2 > \mathbf s$. In the YES case,
$d_{tv}(\mu_1,\mu_2) < \mathbf  s$
and hence for any $t \geq 1$,
\begin{align}
d_{tv}(P^t(x,\cdot),\pi) < \frac{1}{2} \mathbf s < \frac 14 -\delta
\end{align}
In the NO case, $d_{tv}(\mu_1,\mu_2) > \mathbf c $ and hence
\begin{align}
d_{tv}(P^{ct}(x,\cdot),\pi) \geq \frac{1}{2}
 \left(\frac{m-2}{m}\right)^{ct-1} \mathbf c \geq \frac{1}{2}
 \left(\frac{m-2}{m}\right)^{ct-1} .
\end{align}

Since $m \geq 3$, if $ct< \frac{m}{4}
\ln\left(\frac{2}{1+4\delta}\right)$, then $d_{tv}(P^{ct}(x, \cdot),
\pi) > \frac14 +\delta$.
Further, we see that in both the YES and NO case,
\begin{align}
 \tau(1/4) \leq m
\end{align}
We conclude the reduction by setting $t_{\max}=m$.
\end{proof}

Next we prove Theorem~\ref{thm:conp-hard}
and classify the complexity of diagnosing mixing from an arbitrary
starting state given that the chain mixes in polynomial time.
We will use the following result relating mixing time to the
{\em conductance}.

\begin{definition}[Conductance, see e.g. \cite{Sin}]Let $M$ be a Markov chain
  corresponding to the random walk on an edge weighted graph with edge
  weights $\{w_e\}$.
Let $d_x$ denote the weighted degree of a vertex
  $x$. Define
  the {\em conductance} of $M$ to be $\Phi(M) := \min_{\varnothing
    \neq A \subsetneq  \Omega} \Phi_A(M)$ where
\begin{align}
\Phi_A(M) := \frac{\displaystyle\sum_{x \in A, y \in A^c} w_{xy}}
  {\displaystyle\sum_{x \in A} d_x}
\end{align}
\end{definition}

\begin{theorem}[see \cite{Sin}]\label{thm:mixing-conductance}
Let $M$ be a Markov chain
  corresponding to the random walk on an edge weighted graph with edge
  weights $\{w_e\}$ as above. Let $\pi$ be the stationary distribution
  of the Markov chain.
\begin{align*}
 \tau(\eps) \leq \frac{2}{\Phi^2(M)}
 \log\left( \frac{2}{\pi_{\min}\eps}\right)
\end{align*}
where $\pi_{min}$ is the minimum stationary probability of any vertex.
\end{theorem}

\begin{lemma}For every $c\geq 1, 0 < \delta \le 1/4$, {\sc
  GPTC$_{c,\delta}$} is in $\co\am$.
\end{lemma}
\begin{proof} In the first step of the $\co\am$ protocol for {\sc
    GPTC$_{c,\delta}$} the prover sends a pair $x,y \in \Omega$ that maximizes
$d_{tv}(P^t(x,\cdot),P^t(y,\cdot))$. Let $C_x$ be the cicuit which
  outputs the distribution
$P^t(x,\cdot)$ and let $C_y$ output the distribution $P^t(y,\cdot)$.

In the YES case
$\tau(1/4-\delta) <t$ and for every $x,y$,
$d_{tv}(P^t(x,\cdot),P^t(y,\cdot))  < 1/4 -\delta$. In the NO case,
$\tau(1/4+\delta) > ct$ and $c\geq 1$, therefore there must exist
$x,y$ such that
$d_{tv}(P^t(x,\cdot),P^t(y,\cdot))  >  1/4 +\delta$.

By Claim \ref{cl:AM-SD} there is an $\am$ protocol $P$ for {\sc
  SD}$_{1/4+\delta,1/4-\delta}$ with completeness $3/4+\delta$ and
soundness $3/4-\delta$. The prover and the verifier now engage in the
$\am$ protocol to distinguish whether the distance between the two
distributions is large or small. The completeness and soundness follow
from those of the protocol $P$.
\end{proof}

\begin{lemma}Let $0 \le \delta < 1/4$. For
  $1 \leq c < 1/2 \sqrt{\tmax/t^2n^3}(3/4-\delta)$,
  it is $\co\np$-hard to decide {\sc GPTC$_{c,\delta}$}.
\end{lemma}
\begin{proof}
The proof is by reduction from {\sc UnSAT}, which is
  $\co\np$ hard. Let $\psi$ be an
  instance of {\sc UnSAT}, that is, a CNF formula on $n$
  variables. The vertices of the Markov chain are the vertices of the
  hypercube $H$, $V(H) = \{0,1\}^n$ and edges $E(H)  = \{(y_1,y_2):
  |y_1-y_2| =1 \}$.  We set edge weights for the Markov chain as
  follows. Let $d$ be a parameter to be chosen later which is at most
  a constant.
\begin{itemize}
\item For each edge in $E(H)$ set the weight to be 1.
\item If $\psi(y) = 0$ add a self loop of weight $n$ at $y$.
\item If $\psi(y) = 1$ add a self loop of weight $n^{d}$ at $y$.
\end{itemize}

In the YES case, if $ \psi$ is unsatisfiable, the Markov chain is just
the random walk on the
hypercube with probability $1/2$ of self loop at each vertex and it is
well known that
\begin{align*}
\tau(1/4 -\delta) \leq C_{\delta}n \log n
\end{align*}
where $C_{\delta}$ is a constant depending on $(1/4-\delta)^{-1}$ polynomially.

In the NO case, where $\psi$ is satisfiable, we will lower bound the
time to couple from a satisfying
state $y$ and the state $\overline{y}$, obtained by flipping all the
bits of $y$. Consider the distributions $X(t),Y(t)$ of the chain which
are started at $y$ and at $\overline{y}$. We can bound the variation
distance after $t$ steps as follows

\[
d(t) \geq 1- P[\exists s \leq t \ s.t.\  X(s) \neq y] - P[\exists s
  \leq t \ s.t.\  Y(s)=y]
\]
In each step, the chain started at $y$ has chance at most
$1/(n^{d-1}+1)$ of leaving.  On
the other hand, the
probability that the walk started from $\overline y$ hits $y$ in time $t$ is
exponentially small. Therefore

\[
d(t) \geq 1 - 2t/(n^{d-1}+1)
\]
which implies that
\[
\tau(1/4+\delta) >  \frac12 n^{d-1}(3/4-\delta).
\]

Choose $d$ to be a large enough
constant (which may depend polynomially on $\delta^{-1}$), such that
\begin{align}\label{eq:bound-ct}
 \frac12 n^{d-1}(3/4-\delta)  > cC_{\delta} n \log n.
\end{align}

On the other hand we can show a polynomial upper bound on the mixing
time by bounding the conductance as follows. Let $M'$ be the Markov
chain which is the random walk on the hypercube
with self loop probabilities of $1/2$ (where the edge weights are as
in the case where $\psi$ is unsatisfiable). We bound the conductance of $M$
by showing it is not too much smaller than the conductance of $M'$. We
use the fact that for any vertex $x$, the weighted degree $d_x \leq
(n^{d-1}+1) d_x'$.
Let $A \subseteq V(H).$
\begin{eqnarray*}
\Phi_A(M) = \frac{\displaystyle\sum_{x \in A, y \in A^c} w_{xy}}
  {\displaystyle\sum_{A \in \Omega} d_x}
 = \frac{\displaystyle\sum_{x \in A, y \in A^c} w'_{xy}}
  {\displaystyle\sum_{A \in \Omega} d_x}
\geq \frac{\displaystyle\sum_{x \in A, y \in A^c} w'_{xy}}
  {(n^{d-1}+1) \displaystyle\sum_{A \in \Omega} d_x'}
\geq \frac{\Phi_A(M')}{n^{d-1}+1} \geq \frac{1}{n^{d}+n}
\end{eqnarray*}

where we are assuming the lower bound on the conductance of the
hypercube is $\frac{1}{n}$?

We can lower bound $\pi_{min}$ by $1/(n2^{n-1}+n^d 2^n)$ and hence we
have for large enough $n$,
\[
\log (\pi_{min})^{-1} \leq 2n
\]

and hence by Theorem \ref{thm:mixing-conductance}, $\tau(1/4)
\leq 32 n^{2d+1}$.

The reduction can be completed by setting $x = 0^n$, the vector
of all $0$'s, $\tmax=32n^{2d+1}$ and $t = C_\delta n \log n$. By
(\ref{eq:bound-ct}), we see that $ct< 1/2
\sqrt{\tmax/n^3}(3/4-\delta)$ as required.
\end{proof}


\section{Estimating Mixing Time for Arbitrary Markov Chains}
In this section we prove Theorem
\ref{thm:pspace-complete}, saying that the problem of testing
convergence is $\pspace$-complete. The idea of the hardness result is
to simulate
any $\pspace$ computation by a Markov chain so that if there is an
accepting computation, the chain mixes quickly, while if the
computation does not accept then the chain takes much longer to
mix.



We also recall some standard complexity theory background.
\begin{definition}A problem $L$ is in $\pspace$ if there exists a
  Turing machine $M$ which on input $x$ of size $n$ uses a work tape
  with at most a polynomial $p(n)$ number of bits and outputs $M(x) =1$
  iff $x \in L$.
\end{definition}

\begin{definition}
A problem $L_1$ is {\em polynomial time reducible} to another problem
$L_2$ if there exists a polynomial time computable function
(i.e. there is a polynomial time TM which computes the output of $f$)
$f:\{0,1\}^*
\rightarrow \{0,1\}^*$ such that $x \in L_1$ iff $x \in L_2$.
\end{definition}

\begin{definition} A problem $L$ is {\em $\pspace$-hard} if any $A \in
  \pspace$ is polynomial time reducible to it.
\end{definition}

\begin{definition}A problem $L$ is in $BP_H\pspace$ if there is a
probabilistic polynomial space Turing machine $M$ which on input $x$
can flip any number of coins that is a bounded function in $|x|$ (but can only
store polynomially many of them), halts for every
setting of the tosses, and satisfies
\begin{itemize}
\item If $x \in L$ then $P_r(M(x) = 1) \geq 2/3$.
\item If $x \notin L$ then $P_r(M(x) = 1) < 1/3$.
\end{itemize}
\end{definition}

The following result can be deduced from Savitch's Theorem \cite{Sav}.

\begin{theorem}
$BP_H\pspace = \pspace$
\end{theorem}

The proof of Theorem \ref{thm:pspace-complete} now follows by the
following two lemmas. The following lemma uses the fact that the
$t$-step transition probabilities of a Markov chain can be
approximated in $BP_HSPACE$ (see for example \cite{Saks}). We
include all the details here for completeness.

\begin{lemma} For every $1 \le c \le \exp(n^{O(1)})$ and
  $\exp({-n^{O(1)}})< \delta \le 1/4$, the
  problem {\sc GTC$_{c,\delta}$}  is in $BP_H\pspace$.
\end{lemma}

\begin{proof}The proof is by showing that there is a
  randomized algorithm $A$ for {\sc GTC$_{c,\delta}$} with 2-sided
  error using at
  most a polynomial amount of space and exponentially many
  random bits. In particular, the algorithm $A$ on input $X=(C,\hat x,t)$
  queries $C$ at most exponentially many times and
\begin{itemize}
\item If $\tau(1/4 - \delta) \leq t$, $P(A(X,r) = 1) \geq 2/3.$
\item If $\tau(1/4 + \delta) > ct$, $P(A(X,r) = 1) < 1/3.$
\end{itemize}

We show below an algorithm to calculate a $\delta$ additive approximation $\hat
d(t)$ to $d(t)$ with probability at least $2/3$.
The algorithm accepts if $\hat d(t) \leq 1/4$ and rejects otherwise.

In the YES case, $d(t) \leq 1/4 - \delta$ and therefore with
probability at least $2/3$, $\hat d(t) \leq 1/4$ and the
algorithm will accept. In the NO case, $d(ct) >  1/4 + \delta$. Since
the distance $d$ is non-increasing (see e.g. Chapter 4 in
\cite{AlFi}),  $d(t) > 1/4 + \delta$. Therefore, with
probability at least $2/3$, $\hat d(t) > \frac14$ and the algorithm
will reject.

The algorithm to compute $\hat d(t)$ is as follows. Note first that it
is possible to enumerate over all elements of the
state space of the chain $\Omega$ using at most a polynomial amount of
space. It is enough to check for each state whether it is reachable
from $\hat x$ which can be done in $\pspace$ once we can enumerate all
the adjacencies $y$ for a vertex $v$. But this can be done in $\pspace$ by
running over all possible random strings and checking if for some $r$,
$C(v,r) = y$.

%

For $x \in \Omega$ the algorithm runs the chain for $t$ steps
$N$ times, starting at $x$ each
time, and sets $f_{x,z}$ to be the fraction of times the chain stops at
$z$. The estimates $f_{x,z}$ and $f_{y,z}$ can be computed with a polynomial
amount of space in this way. Let $$M_{xy}^t = \frac{1}{2}\sum_{z \in \Omega}
|f_{x,z} - f_{y,z}|.$$ $M_{xy}^t$ can be computed with a polynomial
amount of space by running
over all $z$. Let $$\hat d(t) = \max_{x,y} M_{xy}^t.$$

There are two sources of error in the estimate for $P^t(x,z)$. The
first is due to
using only a polynomial amount of
space, whereas the $t$-step probabilities may be doubly
exponentially small. The size of the error is inversely exponential in
the space we use. This (additive) error can be bounded by
$\delta_a = \delta/4$ using a polynomial amount of space since
$\delta $ is always at least $\exp(-n^{O(1)})$. The second
source of error $\delta_r$ is random and can be bounded by $\delta/4$
by Chernoff bounds for an overall error of at most $\delta/2$.
Thus, if the number of runs $N$ is at least
$48n\delta^{-2}$, by Chernoff bounds, $$P(|P^t(x,z) - f_{x,z}| >
\delta/2 ) \leq 2^{-3n-2}.$$


Therefore, for every $x,y$, taking union over all $z$,
$$P(|M_{xy}^t- d_{tv}(P^t(x,\cdot),P^t(y,\cdot))| > \delta) \leq
2^{n}2^{-3n-2} \leq 2^{-2n-2}.$$

Therefore, we have
\begin{eqnarray*}
P(|\hat d(t) - d(t)| > \delta) \leq P(\exists \ x,y \ \ \text{s.t.} \ \
|M_{xy}^t- d_{tv}(P^t(x,\cdot),P^t(y,\cdot))| > \delta) \leq \frac{1}{4}
\end{eqnarray*}
where the last inequality follows by taking the union over all $x,y$.
\end{proof}

\begin{lemma}
For every  $1 \le c \le \exp(n^{O(1)})$ and $ 0 \leq
\delta < 1/4$, it is  $\pspace$-hard to decide {\sc GTC$_{c,\delta}$}.
\end{lemma}
\begin{remark}
In fact, the conclusion holds even if the Markov chain is restricted
to be reversible. 
\end{remark}
\begin{proof}
We construct a polynomial time reduction
from any $A \in \pspace$ to {\sc GTC$_{c,\delta}$}. Equivalently, a poly-time
computable function $f$
which will map strings $y \in A$ to strings that are YES instances for
{\sc GTC$_{c,\delta}$} and $y \notin A$ to NO instances of {\sc
  GTC$_{c,\delta}$}.

Since $A \in
\pspace$, there a polynomial $n(m) \geq m$ and a Turing machine $M_A$
which on input $y$ of size 
$m$ uses at most $n=n(m)$ space and accepts if and only if $y \in A$.

For each $y \in \{0,1\}^m$, we define $f(y) = (C,x,t)$ as follows.

The state space $\Omega$ of the Markov chain $P$ is the set of all
possible configurations of the machine $M_A$ with input $y \in
\{0,1\}^m$. Since the machine uses space $n$, the state space of the
Markov chain is a subset of $\{0,1\}^n$.

Without loss of generality let $s_{start}$ be the
starting state corresponding to the input $y$, and without loss of
generality, let $s_{acc}$ and $s_{rej}$ be the unique accept and
reject states of the machine. In the case where either $s_{acc}$ or
$s_{rej}$ are never reached, 
they will be added to the state space.
The Markov chain $P$ is a reversible random walk defined by
setting edge weights as follows. There will be two types of weights:
$1$ and $w$ where 
\[
w = \lceil \frac{1000 D^3 c 2^{3 n}}{1-4\delta} \rceil. 
\]

\begin{itemize}
\item There will be a single edge of weight $1$ connecting $s_{rej}$
  and $s_{acc}$. 
\item For each pair of states that are connected by a single step of
  the machine they will be connected 
by a single edge of weight $w$.
\item
There will be an edge of weight $w$ connecting $s_{rej}$ to $s_{start}$.
\item
There will be a loop of weight $w$ connecting each state to itself.
\end{itemize}

A key role in the proof will be played by the graph $G$ of all states
of the machine connected by edges of weight $w$ corresponding to a
single step of the machine.  For any state of the machine, the number
of vertices connected to it 
by a transition of 
  $M_A$ is at most a constant denoted $D-2$ (depending on the finite number of
  states of the machine and a constant number of bits on the tape.)
  This implies in particular that the graph $G$ is of bounded degree $D$.

The circuit $C$ will specify this Markov chain. For a polynomial time
reduction, we require that the
description of the circuit is at most polynomial in $m$. Since $c \leq
\exp(n^{O(1)})$, and $m 
  \leq n$, all the probabilities of the Markov chain can be specified
  by polynomially many bits in $n$ and hence polynomially many bits in
  $m$. Secondly, because the TM
  reads and writes to only a small number of bits, we only have to
  check for a few vertices $v$ whether there is an edge from $u$ to $v$, and
  this can be done with a polynomial sized circuit.

Next, we show bounds on the mixing time with the edge weights as
defined above. For this we observe the following.

\begin{claim}\label{claim:yes-pspace}
In the YES case the graph $G$ is connected and $\tau(1/4-\delta) \leq
10 D^3 2^{3n}/(1-4\delta)$.
\end{claim}

\begin{proof}
Note that by the assumption on the Turing machine, all states are
connected by $w$ edges to either
$s_{acc}$ or $s_{rej}$. Since $s_{rej}$ is connected by a $w$-edge to
$s_{start}$ and $s_{start}$ is connected to
$s_{acc}$ since we are in the YES case it follows that the graph $G$
is connected.

We now use the conductance bound on the mixing time from Theorem
\ref{thm:mixing-conductance} in the following way. For every set but
the empty set or the complete graph, there is weight at least $w$ from
the set to its complement. Furthermore - the total weight of each set
is at most
$D 2^n w$. Therefore the conductance $\Phi \geq D^{-1}2^{-n}$ and hence we
conclude that the mixing time $\tau(\varepsilon)$ is at most
\[
2 D^2 2^{2n} \log(2/\pi_{\min} \eps).
\]
$\pi_{\min}$ is the minimum probability of any state in the space
which we can lower bound by
$w/ (D 2^n w) = D^{-1} 2^{-n}$ and $\eps = 1/4-\delta$. The proof follows since
\[
\log 2 \leq 1, \quad \log(1/\pi_{\min}) \leq \log(D 2^n) \leq
D 2^n, \quad \log(1/\eps) =
\log(1/(1/4-\delta)) \leq 4/(1-4\delta).
\]
\end{proof}

\begin{claim}\label{claim:no-pspace}
In the NO case the graph $G$ is not connected. Moreover $d(t) \geq 1 -
2 t / w$ for all $t$ and
\[
\tau(1/4+\delta) \geq \tau(1/2) \geq w/4.
\]
\end{claim}

\begin{proof}
We first note that the bound $\tau(1/4+\delta) \geq \tau(1/2) \geq
w/4$ immediately follows from $d(t) \geq 1 - 2 t / w$.

In order to show that the graph is not connected we note that
$s_{start}$ and $s_{acc}$ are not in the same component.
This follows from the fact that all edges of $G$ are legal transitions
of the machine. The only other edges of the Markov chain are loops or
the edge connecting $s_{rej}$ to $s_{start}$. Consider in the graph
$G$ the component of $s_{start}$ and of $s_{acc}$ denoted by
$C_{start}$ and $C_{acc}$ respectively.

In order to bound $d(t)$ we look at the distributions $X(t),Y(t)$ of
the chain started at
$s_{start}$ and at $s_{acc}$. Note that
\[
d(t) \geq 1-P[\exists s \leq t \ s.t. \ X(s) \in C_{acc}]-P[\exists s
  \leq t \ s.t. \ Y(s) \in C_{start}]
\]
We note that the only way to move between the components is by
following the edge $1$ weight and the probability of taking this edge
at any step (conditioned on the past) is at most $w^{-1}$. It
therefore follows that: 
\[
d(t) \geq 1-2t/w,
\]
as needed. \end{proof}

By the claims
\ref{claim:yes-pspace} and \ref{claim:no-pspace},
\[
\frac{w/4}{10 D^3 2^{3n}/(1-4\delta)} \geq
\frac{\frac{1000 D^3 c 2^{3 n}}{4(1-4\delta)}}{\frac{10 D^3
    2^{3n}}{1-4\delta}} \geq c.
\]

To complete the reduction,
let $t=10D^32^{3n}/(1-4\delta)$ and set the starting
state $x=s_{start}$.

\end{proof}

\section*{Acknowledgments} The authors would like to thank Salil
Vadhan for helpful discussions.


\begin{thebibliography}{99}



\bibitem{AlFi} D. Aldous and J. Fill. Reversible Markov chains and
  random walks on graphs. Draft at {\tt
  http://www.stat.Berkeley.edu/users/aldous }


\bibitem{AGT} S. Asmussen, P.W. Glynn, H. Thorisson.
  Stationarity detection in the initial transient problem, {\em ACM
  Transactions on Modeling and Computer Simulation}, vol. 2
  no. 2, pp. 130-157, 1992.


\bibitem{Boa} BOA, Bayesian Output Analysis. {\em Available at} {\tt
    http://www.public-health.uiowa.edu/BOA}.

\bibitem{BR}  S. Brooks and G. Roberts. Assessing convergence of
  Markov Chain Monte Carlo algorithms, {\em Statistics and
  Computing}, 8, 319-335, 1998.


\bibitem{CL-book} B. Carlin and T. Louis. Bayes and Empirical Bayes
  methods for data analysis, {\em Chapman and Hall}, 2000.


\bibitem{CC} M. Cowles and B. Carlin. Markov Chain Monte Carlo
  Convergence Diagnostics: A Comparative Review,
  {\em J. Am. Stat. Assoc.} 91, No.434, 883-904, 1996.


\bibitem{GRS-book} W. Gilks, S. Richardson and D. Spiegelhalter
    Ed. Monte Carlo
    Statistical Methods, {\em Chapman and Hall}, 1995.

\bibitem{Der:87}
B. Derrida, G Weisbuch.
Dynamical phase transitions in 3-dimensional spin glasses.
{\em Europhys. Lett}, 4(6), 657-662, 1987.

\bibitem{HR} T. Holenstein and R. Renner. One-Way Secret-Key Agreement
  and Applications to Circuit Polarization and Immunization of
  Public-Key Encryption, In Proceedings of {\em CRYPTO}, 478-493, 2005.


\bibitem{IN-book} I. Ntzoufras. Bayesian Modeling Using WinBUGS, {\em
    Wiley}, 2009.


\bibitem{Jer-book} M. Jerrum. Counting, Sampling and Integrating:
  Algorithms and Complexity, {\em Birkh\"auser}, 2003.


\bibitem{JS} M. Jerrum, A. Sinclair. Polynomial-time Approximation
    Algorithms for the Ising Model, {\em SIAM Journal on Computing},
    22, pp. 1087-1116, 1993.

\bibitem{JSV} M. Jerrum, A. Sinclair, E. Vigoda. A Polynomial-time
  Approximation Algorithm for the Permanent of a Matrix with
  Non-negative Entries, {\em Journal of the ACM}, 51(4):671-697, 2004.




\bibitem{LPW-book} D. Levin, Y. Peres and E. Wilmer. Markov
    Chains and Mixing Times, 2008.

\bibitem{LV03} L. Lov\'asz and S. Vempala. Simulated Annealing in Convex
  Bodies and an $O*(n^4)$ Volume Algorithm, {\em Proc. of the 44th
  IEEE Symposium on Foundations
  of Computer Science}, 2003.

\bibitem{LV06} L. Lov\'asz and S. Vempala. Fast Algorithms for
  Logconcave Functions: Sampling, Rounding, Integration and
  Optimization, {\em Proc. of the 47th IEEE Symposium on Foundations
  of Computer Science}, 2006.



\bibitem{RC-book} C. Roberts and G. Casella. Monte Carlo
    Statistical Methods, {\em Springer}, 2004.


\bibitem{PBCV} M. Plummer, N. Best K. Cowles and K. Vines. CODA:
  Convergence Diagnosis and Output Analysis for MCMC, {\em R News},
  6:1, 7-11, 2006, {\tt http://CRAN.R-project.org/doc/Rnews/}.

\bibitem{Sin} A. Sinclair. Algorithms for Random Generation and
  Counting. {\em Birkhauser}, 1993.

\bibitem{SaVa} A. Sahai, S. Vadhan. A complete promise problem for
    statistical zero-knowledge, {\em Proceedings of the 38th Annual
    Symposium on the Foundations of Computer Science}, 448-457, 1997.

\bibitem{Saks} M. Saks, Randomization and Derandomization in
  Space-bounded Computation. {\em In Proceedings of the 11th Annual
  IEEE Conference on Computational Complexity}, 128-149, 1996.



\bibitem{Sav} W.J. Savitch. Relationships Between Nondeterministic
and Deterministic Space Complexities, {\em J.
Comp. and Syst. Sci.}, 4(2):177-192, 1970.

\end{thebibliography}
\end{document}